\documentclass[journal]{IEEEtran}
\usepackage{cite}
\usepackage{amssymb}
\usepackage{subfigure}
\usepackage{graphics, graphicx,epsfig}  
\usepackage{amsmath}   
\usepackage{comment}

\usepackage{amsfonts,amsthm,amssymb}
\usepackage{psfrag}
\usepackage{ifthen}
\usepackage{bbm}
\usepackage{color}
\usepackage{cite}

\newlength{\widebarargwidth}
\newlength{\widebarargheight}
\newlength{\widebarargdepth}

\newtheorem{theorem}{Theorem}[section]
\newtheorem{lemma}[theorem]{Lemma}

\newtheorem{corollary}[theorem]{Corollary}

\newenvironment{definition}[1][Definition]{\begin{trivlist}
\item[\hskip \labelsep {\bfseries #1}]}{\end{trivlist}}

\begin{document}
\title{Economic information from Smart Meter: Nexus Between Demand Profile and Electricity Retail Price}

\author{Yang~Yu,~\IEEEmembership{Student Member,~IEEE,}
        Guangyi~Liu,~\IEEEmembership{Senior Member,~IEEE,}
        Wendong~Zhu,
        Fei~Wang,
        Bin~Shu,
        Kai~Zhang,
        Ram~Rajagopal,~\IEEEmembership{Member,~IEEE,}
        and~Nicolas~Astier.
\thanks{Yang Yu is with Precourt Energy Efficiency Center and Department of Management Science and Engineering, Stanford University, CA, 94305,yangyu1@stanford.edu.}
\thanks{Guangyi Liu, Wendong Zhu, and Fei Wang are with GEIRI North America, 5451 Great America Parkway, Santa Clara, CA 95054, guangyi.liu@geirina.net.}
\thanks{Bin Shu and Kai Wang are with Beijing Electric power Economic  Technology Research Institute, Beijing, China, 100055.}
\thanks{Ram Rajagopal is with Department of Civil and Environmental Engineering and Department of Electronic Engineering, Stanford University, CA, 94305,ramr@stanford.edu.}
\thanks{Nicolas Astier is with Toulouse School of Economics, Toulouse, France, 31500,nicolas.astier@ut-capitole.fr.}
}

\maketitle

\begin{abstract}
In this paper, we demonstrate that a consumer's marginal system impact is only determined by their demand profile rather than their demand level. Demand profile clustering is identical to cluster consumers according to their marginal impacts on system costs. A profile-based uniform-rate price is economically efficient as real-time pricing. We develop a criteria system to evaluate the economic efficiency of an implemented retail price scheme in a distribution system by comparing profile clustering and daily-average clustering. Our criteria system can examine the extent of a retail price scheme's inefficiency even without information about the distribution system's daily cost structure. We analyze data from a real distribution system in China. In this system, targeting each consumer's high-impact days is more efficient than target high-impact consumers.
\end{abstract}



\section{Introduction}
Consumers' daily demand profiles fundamentally links with the cost of serving them in a distribution network system (DS)\cite{strbac2008demand}. A consumer's daily demand profile is their hourly electricity demand across a day, measured from a smart meter. The aggregated demand profile in a DS determines the DS's daily system cost and associated emissions\cite{madaeni2013using}. Most current retail price designs are based on demand levels rather than demand profiles. For example, a retail electricity bill in most current U.S. markets includes two parts: a tiered price for electricity energy consumption, and a fixed charge to balance the utility's budget for fixed costs and other service costs\cite{puller2013efficient}. Both parts are independent of a consumer's demand profile.

The lack of information about individual users' electricity demand patterns limits retail price designs and demand response project targets. Before smart meters collect individual consumers' demand profiles, retail-price designs and demand response projects are designed on the basis of two assumptions\cite{joskow2007reliability}. One is that all consumers have very similar electricity usage patterns. The other is that the supply cost curve is convex. These two assumptions lead to the conclusion that the marginal system cost of serving a consumer monotonically increases with their demand level. For example, the tiered-price scheme charges consumers whose monthly total demand levels exceed a given threshold at a higher rate than other users. The demand charge also prices consumers only according to their peak demand levels in a month no matter when their peak demands occur. These "demand-level based" policies are efficient only if the assumption that consumers' usage patterns are homogeneous is true. Unfortunately, this assumption is rejected by the empirical data analyses on smart meter data.

Researchers have found remarkable heterogeneity in consumers' daily demand profiles in transmission systems \cite{wang2015load}. \cite{kwac2014household} demonstrates that annual consumers' demand profiles in PG\&E can be clustered into more than 200 types. The heterogeneity in consumers' daily-demand profiles suggests that we should analyze the difference in marginal costs in serving various consumers, examine the relation between consumers' marginal system impacts and their demand profiles and levels, and reconsider our current retail-pricing scheme. 

Consumer daily demand profiles provide our information about features impacting consumption behaviors, such as temperature\cite{albert2013smart}\cite{albert2013building}\cite{kwac2013utility}\cite{beckel2014revealing}\cite{wijaya2014consumer}. Demand profiles are also used for improving the accuracy of demand forecasts\cite{quilumba2015using}\cite{rhodes2014clustering}\cite{gulbinas2015segmentation}\cite{kouzelis2015estimation}. A lot of studies try to target consumers for demand management on the basis of demand profiles \cite{patel2014aggregation}\cite{kavousian2015ranking}\cite{borgeson2013targeted}\cite{verma2015data}.

However, we still lack an economic explanation for consumers' heterogeneous but clustered demand profiles. Consequently, it is unclear how the profile clustering results can be used for distribution system market designs and demand response project developments. Two core questions must be answered before consumer profile information can be fully used to design retail contracts and demand response projects. How do a consumer's demand level and profile impact the distribution system's cost of serving this consumer? What is the relation between consumers' profile clustering and their impacts on a distribution system's costs?

Understanding the relation between consumer demand profile and the system cost is the foundation for the framework to analyze the profitability of serving a consumer, assess the efficiency of the retail price design or demand response project, and examine the quality of demand-response resources in a DS. In this paper, we explain the economic nexus between consumer's demand profile, marginal system cost, and economically optimal daily average retail rate. This economic explanation is also used to develop criteria to assess the economic inefficiency and association impacts of an implemented retail price scheme in a DS. The rest of the paper is organized as follow: Sec~\ref{sec:marginalimpact} includes the theoretical discussion about consumer demand profile and marginal system impacts; we explain the nexus between demand profile clustering and daily-average-rate clustering in Sec~\ref{sec:twoclustering}; then, we develop criteria to assess a retail-price design in Sec~\ref{sec:criteriaforprice}; Sec~\ref{sec:empirical} summarize our empirical analysis on a China's DS; Sec~\ref{sec:conclusion} includes a discussion about the empirical studies and concludes the whole paper.

\section{Load profile clustering, marginal system impact and pricing design}\label{sec:marginalimpact}

\subsection{A consumer's marginal system impacts}

We examine a DS that serves $N$ consumers whose hourly demands are positive. Consumer $i$'s daily demand is denoted by  $\overrightarrow{L_i} = (l_{i1}, ..., l_{ih}, ...l_{i24})_{h=1}^{24} \in R^{24}$, where  $l_{ih}$ is consumer $i$'s hourly demand in hour $h$.

The summation of all consumers' demands is the DS's daily system aggregated demand $\overrightarrow{L} = \sum_{i=1}^N \overrightarrow{L_i}$, which determines the DS's daily system cost and emissions. The daily cost $C(\overrightarrow{L} )$, as well as emissions, is determined by not only the total amount of electricity demand but also by various features of $\overrightarrow{L}$'s profile, including peak level, peak time and ranges of ramps. For example, $C(\overrightarrow{L})$ is high when $\overrightarrow{L}$ has a big evening ramp, such as California's "duck-curve" concerns. We assume $\overrightarrow{L}$ has $M$ features that impact the system cost and use $\phi_j$ to represent the $j$th feature. The feature $\phi_j$ is a function of $\overrightarrow{L}$. In this section, we examine the marginal impact of a consumer's daily demand on the daily cost $C(\overrightarrow{L} )$ through the impact on $\overrightarrow{L}$'s profile features.

$\overrightarrow{L}$ and its features change when consumer $i$ increases their demand and keeps their daily demand profile the same, which means that this additional electricity demand is proportionally allocated into $24$ hours according to $i$'s daily demand profile. $\overrightarrow{L}$' features determine the daily system cost. Thus, Consumer $i$'s profile-keeping demand growth will change the DS's daily system cost. The marginal system daily cost caused by consumer $i$'s daily electricity demand is the derivative of the system daily cost along $\overrightarrow{L_i}$. 

We define $i$'s marginal system feature impacts (MFIs) and marginal system cost impact (MCI) as the marginal changes of $\overrightarrow{L}$'s profile features and system costs corresponding to $\overrightarrow{L_i}$'s profile-keeping demand growth.
\begin{definition}\label{def:sensitivity}
Consumer $i$'s marginal system feature impact on $\phi_j$ is
\begin{align}
&MFI_{i,j} = \nabla_{\overrightarrow{L_i}} \phi(\overrightarrow{L}) \nonumber \\
&= \lim_{\Delta \rightarrow 0} \frac{\phi_j(\overrightarrow{L}+\Delta\frac{\overrightarrow{L_i}}{\parallel \overrightarrow{L_i} \parallel_1}) - \phi_j(\overrightarrow{L})}{\Delta} \label{ref:margphi}.
\end{align}
Here, $\parallel\bullet\parallel_1$ is the $l_1$ norm of $R^{24}$. Correspondingly, the marginal system-cost impact of consumer $i$'s daily demand is
\begin{align}
MCI_i &= \nabla_{\overrightarrow{L_i}} C(\overrightarrow{L})\nonumber \\
&=\sum_{j=1}^M (\frac{\partial C}{\partial \phi_j} MFI_{i,j}). \label{ref:margico}
\end{align}
\end{definition}
Here, $\frac{\partial C}{\partial \phi_j}$ is feature $\phi_j$'s marginal system-cost impact.

MFIs and MCI are sensitive to the demand profile rather than demand level. We refer two daily demands $\overrightarrow{L_i}$ and $\overrightarrow{L_k}$ share the same demand profile if they are linearly correlated. If two consumers' daily demands share the same demand profile, they have the same MFIs and MCI even if  they consume different amounts of electricity.
\begin{theorem}\label{the:clustersystemimpact}
If $\overrightarrow{L_i}$ and $\overrightarrow{L_k}$ are linearly correlated, consumers $i$ and $k$ share the same MFIs and MCI.
\end{theorem}
\begin{proof}
Assume consumers $i$ and $k$ have the same demand profile. Their daily demands are linearly correlated such that $\alpha\overrightarrow{L_i} = \overrightarrow{L_k}$, where $\alpha$ is a positive constant. $\forall j$,
\begin{align}
 MFI_{i,j} &= \lim_{\Delta \rightarrow 0} \frac{\phi_j(\overrightarrow{L}+\Delta\frac{\overrightarrow{L_i}}{\parallel \overrightarrow{L_i} \parallel_1})-\phi_j(\overrightarrow{L})}{\Delta}.\nonumber \\
& = \lim_{\Delta \rightarrow 0} \frac{\phi_j(\overrightarrow{L}+\Delta\frac{\alpha\overrightarrow{L_i}}{\parallel \alpha \overrightarrow{L_i} \parallel_1})-\phi_j(\overrightarrow{L})}{\Delta}.\nonumber \\
& = \lim_{\Delta \rightarrow 0} \frac{\phi_j(\overrightarrow{L}+\Delta\frac{\overrightarrow{L_k}}{\parallel \overrightarrow{L_k} \parallel_1})-\phi_j(\overrightarrow{L})}{\Delta}.\nonumber \\
& =  MFI_{k,j}.\nonumber
\end{align}
Because $MFI_{i,j} = MFI_{i,k}$ for all $j$, we conclude that $MCI_i = MCI_k$.
\end{proof}

Theorem~\ref{the:clustersystemimpact} is the economic explanation for consumer profile clustering.

\begin{lemma}\label{lem:explaprofileclustering}
Demand profile clustering actually clusters consumers according to their MFIs and MCI. Consumers in the same profile cluster share the same MFIs and MCI.
\end{lemma}

\subsection{Consumer demand profile and economic retail price}
Consumer MCI is fundamentally correlated with the economically optimal retail price. A retail price design that charging consumers according to their MCIs is economically optimal. If consumers pay their daily bill by a rate equal to $MCI_i$, we refer this price scheme as "\textbf{profile price menu}" (PPM). PPM actually proposes consumers a menu of pairs (profile, daily flat price) and lets consumers select demand profile from the menu. PPM is equivalent to charging consumers by real-time hourly prices (RTPs), in two ways. First, PPM gives each consumer the same economic incentive as the RTPs. Second, PPM and RTPs lead to the same market equilibrium in the DS, which means consumer demands and system costs and emissions are the same. In the following theorem, we demonstrate the foundation 
\begin{theorem}\label{the:ppisoptimal}
Charging consumers according to the PPM provides every consumer the same incentive in all hours and consequently leads to the same market equilibrium as charging consumers the RTPs in a DS system. Therefore, PPM is also the economically efficient retail price.
\end{theorem}
\begin{proof}
Our demonstration includes three steps. First, we build the linkage between consumer's MCI and the RTPs. Then, we demonstrate that the PPM and RTP provide the same incentives to every consumer when they decide their hourly demand level. Finally, we demonstrate that the same individual economic incentives lead to the same market equilibrium.

Step 1: We use $\lambda_h$ to represent the wholesale RTP in hour $h$ and $\overrightarrow{\Lambda} = (\lambda_1...\lambda_{24})$ to represent the RTPs in the whole day. The daily system cost is $\overrightarrow{\Lambda} \cdot \overrightarrow{L}^\intercal$. Thus, 
consumer $i$'s MCI
\begin{align}
 MCI_i &= \nabla_{\overrightarrow{L_i}} C(\overrightarrow{L}) = \nabla_{\overrightarrow{L_i}} \overrightarrow{\Lambda} \cdot \overrightarrow{L}^\intercal \nonumber\\
  &= \overrightarrow{\Lambda} \nabla_{\overrightarrow{L_i}} \cdot \overrightarrow{L}^\intercal = \sum_{h=1}^{24} \frac{l_{i,h}}{ \sum_{h=1}^{24} l_{i,h}}\lambda_h. \label{equ:MCIandRTP}
\end{align}
Here, $l_{i,h}$ is consumer $i$'s consumption during hour $h$. Thus, MCI is a weighted average of the RTPs, where the weight for each hour's price is the proportion of the hourly demand to the daily total demand. 

Step 2: We use $U_i(l_{i,1},...,l_{i,24})$ to represent consumer $i$'s utility function. Given the PPM, consumer $i$ pays $MCI_i(l_{i,1},...,l_{i,24})$ for demanding one unit of electricity. $MCI_i(l_{i,1},...,l_{i,24})$ is a function of $l_{i,h}$. Thus, consumer $i$'s hourly demands are solved from the following welfare maximization problem
\begin{align}
 \max_{(l_{i,h},\forall i)} &U_i(l_{i,1},...,l_{i,24}) - MCI_i(l_{i,1},...,l_{i,24}) \times \sum_{m=1}^{24} l_{i,m} \label{equ:consumerppoptimal}
\end{align}
Therefore, the first order conditions are
\begin{align}
 \frac{\partial U_i(l_{i,1},...,l_{i,24})}{\partial l_{i,h}} =\underbrace{ \frac{\partial MCI_i}{\partial l_{i,h}}\times\sum_{m=1}^{24} l_{i,m}}_{A} + \underbrace{MCI_i
 \vphantom{\frac{\partial MCI_i}{\partial l_{i,h}}\times\sum_{m=1}^{24}} 
 }_{B},\forall h. \label{equ:PPfirstordercondition}
\end{align}
The left hand side of Eq.~\eqref{equ:PPfirstordercondition} is consumer $i$'s marginal utility of demanding $l_{i,h}$ in hour $h$, which includes two parts. Term A is the marginal impact of $l_{i,h}$ on all hourly costs by affecting the rate $MCI_i$. Term B is the marginal impact of $l_{i,h}$ on the cost for hour $h$ by varying the demand level. The summation of Term A and B is the RTP for hour $h$.
\begin{align}
 &\frac{\partial MCI_i}{\partial l_{i,h}}\times\sum_{m=1}^{24} l_{i,m}  + MCI_i. \nonumber\\
 =& \frac{\lambda_h\sum_{m \neq h}l_{i,m} - \sum_{m \neq h} \lambda_ml_{i,m}}{(\sum_{m=1}^{24} l_{i,m}) ^2}\times\sum_{m=1}^{24} l_{i,m}   \nonumber\\
 +& \sum_{m=1}^{24} \frac{l_{i,m}}{ \sum_{m=1}^{24} l_{i,m}}\lambda_m = \lambda_h
\end{align}
Therefore, Eq.~\eqref{equ:PPfirstordercondition} is equivalent to
\begin{align}
 \frac{\partial U_i(l_{i,1},...,l_{i,24})}{\partial l_{i,h}} =\lambda_h,\forall h. \label{equ:RTPfirstordercondition}
\end{align}
When the RTP is implemented in the DS, consumer $i$'s welfare maximization problem is
 \begin{align}
 \max_{(l_{i,h},\forall i)} U_i(l_{i,1},...,l_{i,24}) - \sum_{h=1}^{24} \lambda_h l_{i,h} \label{equ:consumerppoptimalRTP}
\end{align}
The first order conditions of Eq.~\eqref{equ:consumerppoptimalRTP} are also Eq.~\eqref{equ:RTPfirstordercondition}. Thus, the PPM and RTP provide the same economic incentives to consumer $i$ in all hours.

Step 3: When all consumers are price takers, the RTP in hour $h$ $\lambda_h(\overrightarrow{L})$ is a function of only the aggregate demand, and is not affected by individual consumer demands. Therefore, the supply cost of the whole DS is $C(\overrightarrow{L}) = \sum_{h=1}^{24} \lambda_h(\overrightarrow{L})\sum_{i=1}{N}$. Consequently, for both PPM and RTP scenarios, the market equilibrium of all hours are solved from
 \begin{align}
 \frac{\partial U_i(l_{i,1},...,l_{i,24})}{\partial l_{i,h}} =\lambda_h(\overrightarrow{L}),\forall i, h. \label{equ:RTPfirstorderconditionmarketeqbuilirium}
\end{align}
Therefore, the PPM and RTPs lead to the same market equilibrium. Because the RTPs are the economically optimal retail price, the PPM is also economically optimal.
\end{proof}

\section{Profile clustering and daily-average rate clustering}\label{sec:twoclustering}
The economic optimality of the PPM scheme indicates the economic essence of consumer demand profile clustering. Theorems~\ref{the:clustersystemimpact} and~\ref{the:ppisoptimal} together indicate that consumers sharing the same demand profile pay the same daily rate under the PPM scheme. Therefore, consumers in the same demand-profile cluster pay the same rate when the PPM is implemented. Because of the economic linkage between the PPM daily rate and RTPs demonstrated in the proof of Theorems~\ref{the:ppisoptimal}, consumers in the same demand-profile cluster also pay the same daily average rate when the RTPs are used.
\begin{corollary}\label{cor:pprtpsamerate}
Consumers in the same demand-profile cluster pay the same daily average rate when the RTPs or PPM is implemented.
\end{corollary}

The conclusion in Corollary~\ref{cor:pprtpsamerate} can be generalized to any economically optimal price scheme in the following theorem, which is the core conclusion of this paper. In this theorem, we conclude that an implemented retail price is economically optimal if and only if every consumer's average daily rate calculated according to the implemented retail price is their MCI. Consequently, if an economically optimal price scheme is implemented,  the partition of consumers according to consumer demand profiles is a refinement of the partition of consumers according to daily average rates. When consumers in two clusterings have different MCIs, the two partitions must be the same.   
\begin{theorem}\label{the:sameprofilesameoptimaldailyrate}
A price scheme is economically optimal if and only if every consumer's daily average rate is their MCI. Therefore, an economically optimal price will lead consumers who share the same demand profile to pay the same daily average rate even if they have different demand levels.
\end{theorem}
\begin{proof}
Theorem~\ref{the:ppisoptimal} has proved the sufficiency. Here, we prove the necessarily. 

We assume that $\Gamma$ is an economically optimal price scheme. The average daily rate of consumer $i$ is $\gamma_i$, which is a function of consumer $i$'s hourly demands and can differ between consumers. Consequently, consumer $i$ chooses their hourly demands under the price scheme $\Gamma$ by solving from the following welfare optimization problem.
\begin{align}
 \max_{(l_{i,h},\forall i)} U_i(l_{i,1},...,l_{i,24}) - \gamma_i(l_{i,1},...,l_{i,24})\sum_{m=1}^{24} l_{i,m} \label{equ:consumerppoptimalunderomega}
\end{align}
Then, consumer $i$'s hourly demand bundle is solved from the first order conditions.
\begin{align}
 \frac{\partial U_i(l_{i,1},...,l_{i,24})}{\partial l_{i,h}} = \frac{\partial \gamma_i}{\partial l_{i,h}} \sum_{m=1}^{24} l_{i,m}  + \gamma_i,\forall h. \label{equ:PPfirstorderconditionunderomega}
\end{align}

Because $\gamma$ is economically optimal, $\frac{\partial U_i(l_{i,1},...,l_{i,24})}{\partial l_{i,h}}$ must be equal to $\lambda_h$. Consequently, according to Cauchy-Kovalevskaya Theorem, the Eq.~\eqref{equ:PPfirstorderconditionunderomega} has a unique analytic solution near $(\lambda_1,...,\lambda_{24})$. Because $MCI_i$ is a solution of Eq.~\eqref{equ:PPfirstorderconditionunderomega}, $\gamma_i$ must be equal to $MCI_i$.
\end{proof}

\section{Criteria to assess the performance of an implemented retail price scheme in a DS}\label{sec:criteriaforprice}
The economic linkages connecting the consumer demand profile, the MCI, and the economically optimal price scheme are useful to assess the economic performance of an implemented retail price scheme in a DS. We can directly induce the most important corollary of this paper from Theorem~\ref{the:sameprofilesameoptimaldailyrate}, which allows us to develop a sequence of criteria to assess how different a current price scheme is from the economically optimal price scheme.
\begin{corollary}\label{cor:samepartition}
When the implemented price scheme in a DS is economically optimal, the partition of consumers according to their demand profile must be a refinement of the partition of consumers according to their daily average rate. If consumers with two demand profiles always have different MCIs, those two partitions must be the same.
\end{corollary}
\begin{proof}
Assume consumer $i$ and $k$ belong to the same demand profile cluster $\omega_p$, we demonstrate that $i$ and $k$ also must belong to the same daily-average-rate cluster $\omega_r$ when the implemented retail price is economically optimal. Theorem~\ref{the:clustersystemimpact} demonstrates that
\[
i,k \in \omega_p \Rightarrow MCI_i = MCI_k
\]

If $i$ and $k$ pay different daily average rate $r_i \neq r_k$, we conclude either $r_i \neq MCI_i$ or $r_k \neq MCI_k$, which is contradictory with the assumption that the implmented price is economically optimal according to Theorem~\ref{the:sameprofilesameoptimaldailyrate}. Thus, $i$ and $k$ must pay the same daily average rate and belong to the same daily-average-rate cluster $\omega_r$
\end{proof}

Corollary~\ref{cor:samepartition} allows us assess the economic efficiency to a large extent even if the information about the RTPs is completely unknown.

When consumers with two demand profiles always have different MCIs, the mismatch of the partition of consumers according to demand profiles denoted by $\Omega_p$ and the partition of averagely daily rates $\Omega_r$ is a measurement of the efficiency of an implemented retail price scheme. The partition $\Omega_p$ must be identical to the partition $\Omega_r$. Otherwise, the implemented price must not be economically optimal. Simultaneously, the difference between $\Omega_p$ and $\Omega_r$ indicates how different is the implemented price scheme from an economically optimal price scheme. The larger the difference, the more inefficient the implemented price scheme.  

When consumer's MCI is injective function of their demand profile, we define the "the degree of consistency between the implemented price scheme and an economically optimal price design" as the difference between $\Omega_p$ and $\Omega_r$. In the DS discussed in above two sections, we assume that consumers are clustered into $T$ profile types and $S$ daily-average-rate types. We use $t$ to represent the $t$th type of demand-profile type $\omega_{p,t} \in \Omega_p$ and $s$ to represent the $s$th type of daily-average-rate type $\omega_{r,s} \in \Omega_r$.

\begin{definition}
The \emph{the degree of consistency} between the implemented price and an economically optimal price in the DS is
\begin{align}
DOC=&\frac{2 \times \sum_{t=1}^{T}\sum_{s=1}^{S} \rho_{ts}\log_2\frac{ \rho_{ts}}{\rho_{t}\rho_{s}}}{-\sum_{t=1}^{T} \rho_{t}\log_2\rho_{t} - \sum_{s=1}^{S} \rho_{s}\log_2\rho_{s}}, \nonumber\\ &\text{ if $ \rho_{t}\rho_{s} < 1$} \nonumber \\
= & 1, \text{otherwise} \label{equ:degreeofconfusion}
\end{align}
Here,
\[
\rho_t = \frac{|\omega_{p,t}|}{N}, \rho_s = \frac{|\omega_{r,s}|}{N}, \text{and  }\rho_{ts} = \frac{|\omega_{p,t} \cap \omega_{r,s}|}{N}.
\]
$|\omega_{p,t}|$ is the number of consumers having $t$th demand profile. $|\omega_{r,s}|$ is the number of consumers paying $s$th daily average rate. $|\omega_{p,t} \cap \omega_{r,s}|$ is the number of consumers that have $t$th demand profile and pay $s$th daily average rate.
\end{definition}
$DOC$ is the normalized mutual information between $\Omega_p$ and $\Omega_r$. $DOC$ measures how much the diversity of a consumer's daily average rate is reduced by when knowing this consumer's daily profile. When the implemented price scheme is economically optimal, $DOC$ is equal to $1$, which means that $\Omega_p$ and $\Omega_r$ are identical. $DOC$ is equal to $0$ if $\Omega_p$ and $\Omega_r$ are independent of each other, which is the worst case. The higher the value of $DOC$, the more efficient the implemented retail price scheme for the DS. 

It is rare that consumers with different profiles have the same MCI. Therefore, $DOC$ can be broadly used for various DSs and markets.

$DOC$ can measure the extent of efficiency of the implemented retail price even if we have no information about the system daily cost $C(\overrightarrow{L})$. When more information about $C(\overrightarrow{L})$ is revealed, further measurements can be established to assess the economic efficiency and impacts in more detail.

If we know that $C(\overrightarrow{L} )$ is monotonically increasing with a linear combination of $\overrightarrow{L}$'s $M$ features $\{\phi_j,j=1...M\}$, we can measure the economic distortion caused by the implemented retail price even if we do not know the exact form of $C(\overrightarrow{L})$. Assume that system cost $C(\overrightarrow{L} )$ is monotonically increased with $\sum_{j=1}^{M}\mu_j\phi_j$, where $\mu_j > 0$. Consumer $i$'s MCI is monotonically increased with $\Phi_i = \sum_{j=1}^{M}\mu_j MFI_{i,j}$. When a economically optimal price scheme is implemented, consumer $i$ pays a higher daily rate than consumer $j$ if and only if $\Phi_i > \Phi_j$. Therefore, we call $\Phi_i$ as the "MCI index". We can sort consumers in a list according to the ascending order of their MCI index. If this list is different from the list of consumers that is sorted in the ascending order of daily average rates $r_i$, then the implemented retail price is not economically optimal. The more different the two lists, the more distortion caused by the implemented retail price. We use the KenDall tau distance to measure the extent of distortion caused by an implemented retailed market price.
\begin{definition}
We define the \textit{distortion} caused by the implemented retail price
\begin{align}
Dt = \frac{|\{(\Phi_i,\Phi_j)|\Phi_i > \Phi_j \wedge r_i < r_j\}|}{N(N-1)/2} \times 100\%\label{equ:distortion}.
\end{align}
\end{definition}

If we have complete information about the function of $C(\overrightarrow{L} )$, such as the local marginal pices. we can accurately measure the inefficiency of the implemented retail price. Consumer $i$ is subsidized by the implemented price scheme if $MCI_i > \lambda_i$ and taxed if $MCI_i < \lambda_i$. The amount of money subsidized to consumer $i$ is $(MCI_i - \lambda_i)\times \sum_{h=1}^{24}l_{i,h}$.

\section{Empirical analysis on a Chinese distribution system}\label{sec:empirical}
\subsection{Data background and processing}
We use real smart meter data from an industrial county of one of the largest cities in China to analyze consumer marginal system impacts and assess the effectiveness of the current retail price scheme in this distribution system. In this city, there are $2,110$ consumers, all of whom are equipped with smart meters. The data ranges from January 2014 to December 2014 and consists of the electricity consumption of commercial, industrial and residential customers measured at 1-hour intervals. The total number of daily demand profiles is $516,494$. 

This county is representative of a Chinese distribution system. In contrast with the United States, China's electricity demand is mostly determined by large commercial and industrial consumers. $70$\% of China's electricity is consumed by industrial and commercial sectors\cite{LBNLChina2014}. In addition, most distribution systems in China usually simultaneously serve industrial, commercial, and residential consumers. In contrast with residential consumers, large industrial and commercial consumers are easier to be managed individually and priced differently. In the distribution system studied in this paper, the majority of consumers are industrial and commercial companies. But there are still a significant number of residential and public-service users. 

We normalized daily demand profiles and use an adaptive $k$-means algorithm to cluster all normalized daily demand profiles into 36 profile types using the elbow method [2]. In each cluster, the euclidean distance from a demand profile to the cluster kernel is less than $5$\% of the $l_2$ norm of the cluster kernel\cite{kodinariya2013review}\cite{morse2007efficient}.

We calculate every consumers' MFIs for the system-aggregated daily demand profile's morning ramp range, evening ramp range, and peak demand level, which are represented by $MFI_{MR}$, $MFI_{ER}$, and $MFI_{PD}$). China does not have a wholesale electricity market and has a fully-regulated wholesale price scheme. Consequently, the system cost is determined mainly by the three features we analyzed. According to Theorem~\ref{the:clustersystemimpact}, we only need to calculate each profile type's MFIs instead of calculating each person's MFIs. Therefore, the computing load is significantly reduced. 

\subsection{Consumers' marginal system impacts}
MFIs provide us the information about whose consumption were more responsible for the aggregate load profile's ramps or peak demands. For a given feature, consumers' MFIs are significantly various. In the Appendix, we present the histograms of $MFI_{MR}$, $MFI_{ER}$, and $MFI_{PD}$. 

Some consumers have negative $MFI_{MR}$ and $MFI_{ER}$ on particular days. Those consumers' consumption are very valuable in moderating aggregate load's morning or evening ramps on those days. If those consumers demand more electricity and keep their demand profiles, the aggregate demand will lead to more moderate morning or evening ramps. For instance, Type 6's demand was ramping down when the aggregate demand was experiencing the morning ramp on April 18th, 2014 (Fig.~\ref{fig:negativeMFI}). Consequently, Type 6 had a negative $MFI_{MR}$ in this day.
\begin{figure}
\centering
\subfigure[System aggregated demand profile]{
\includegraphics[width = 0.3\textwidth,height = 1.8in]{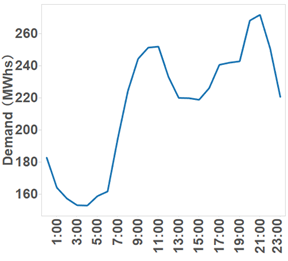}
\label{fig:aggregatedemand}
}
\subfigure[Aggregated demand profile of Type-6 consumers]{
\includegraphics[width = 0.3\textwidth,height = 1.8in]{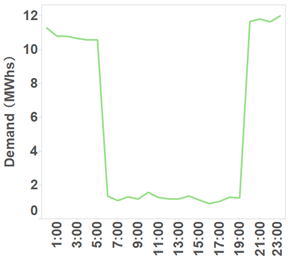}
\label{fig:type6}
}
\caption{Type-6 consumer's $MFI_{MR}$ is negative on April 18th, 2014}
\label{fig:negativeMFI}
\end{figure}

Consumer demands affect the distribution system's cost through different dynamics. The costs serving some consumers are mainly caused by dealing with their morning ramps while the cost serving other consumers are mainly due to their impacts on system's evening ramps or peak demands. In Fig. \ref{fig:triangle}, we plot the MFIs of two types of consumers on March 24th. System costs for dealing with the morning ramps are more due to Type 16's rather than Type 20's consumption. System costs for dealing with the evening ramps and peak loads are more due to Type 20's rather than Type 16's consumption.
\begin{figure}
\centering
\includegraphics[width = 0.33\textwidth]{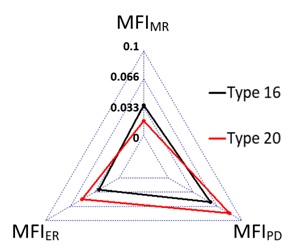}
\caption{Type-16 and Type-20 consumer's MFIs for three features}
\label{fig:triangle}
\end{figure}

Understanding the long-term distributions of a consumer's MFIs is necessary to examine how expensive it is to serve this consumer and where the expense comes from. The long-term distributions of consumers are also useful for planning a distribution system's facilities and assets, such as the size of transformers. We summarize the means and standard deviations of every type of consumer MFIs for three features in Fig. ~\ref{fig:boxplot}. Serving some consumers is less expensive on average than serving others. For example, Type 1 consumers have lower average MFIs for all three features than Type 3 consumers, so serving Type 1 is less expensive than serving Type 3.
\begin{figure}
\centering
\subfigure[Distributions of $MFI_{MR}$ in each profile type]{
\includegraphics[width = 0.45\textwidth,height = 2in]{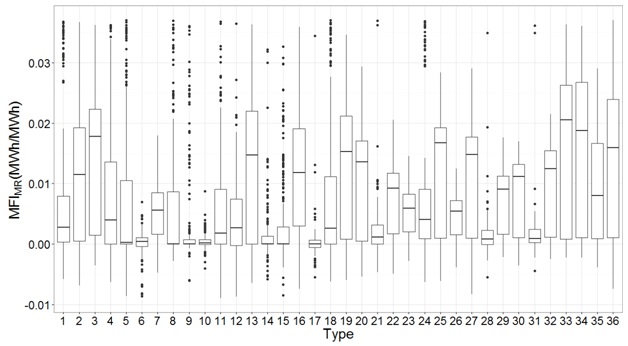}
\label{fig:MFIMRbox}
}
\subfigure[Distributions of $MFI_{ER}$ in each profile type]{
\includegraphics[width = 0.45\textwidth,height = 2in]{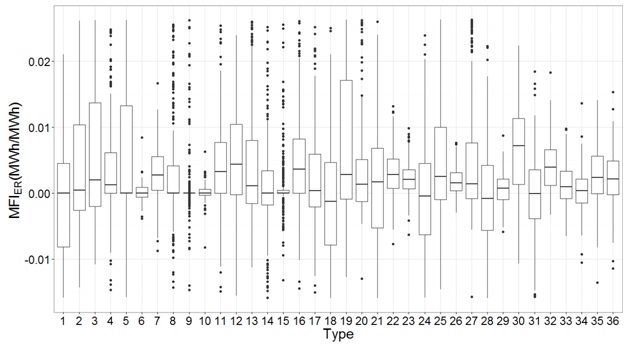}
\label{fig:MFIERbox}
}
\subfigure[Distributions of $MFI_{PD}$ in each profile type]{
\includegraphics[width = 0.45\textwidth,height = 2in]{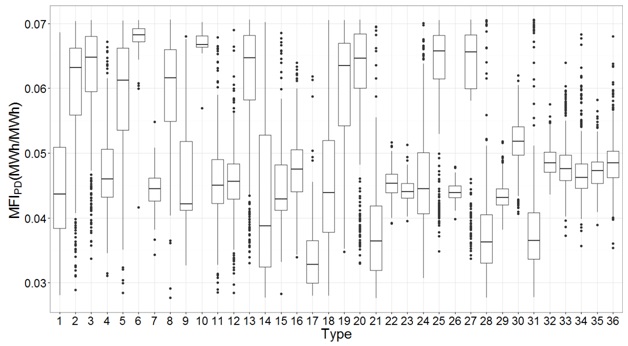}
\label{fig:MFIPDbox}
}
\caption{Variations of three MFIs within and between consumer profile types}
\label{fig:boxplot}
\end{figure}

Many consumer-profile types have MFIs with significant fluctuations. For example, the mean of Type 3's $MFI_{MR}$ is nearly 10 times larger than its $25$\% quantile value. $MFI_{PD}$ has smaller within-type variation than $MFI_{MR}$ and $MFI_{ER}$. The heterogeneities of $MFI_{MR}$ and $MFI_{ER}$ are mainly within types rather than between types. In contrast, $MFI_{PD}$ has more significant cross-type heterogeneity than $MFI_{MR}$ and $MFI_{ER}$. The significant within-type variation of MFIs is caused by the big difference between daily aggregated demand levels and profiles. If two days have different aggregated demand levels or profiles, the same profile type has different MFIs of the same given feature in these two days. 

We also note that the variation of consumer MFIs are positively correlated with the absolute value of the mean of their MFIs. Consumers that have significantly higher or lower average marginal system impacts usually also have significantly varying system impacts. In the Appendix, we provide the scatter plots of average MFIs and the associated standard deviations. 

We note that nearly all profile types have both negative and extreme high positive $MFI_{MR}$ and $MFI_{ER}$. Therefore, every type of consumers can be valuable in moderating system ramps in some "particular days" while their demands significantly aggravate system ramps in some other days. For example, Type 33 had the highest average $MFI_{MR}$. However, this type still had negative $MFI_{MR}$ in $5$\% of days.

\subsection{Consumer demand level and marginal system impacts}
Besides implementing "demand-level based" prices, China's energy management has a long-term principle that is "manage large consumers only". Consequently, most demand-response projects in China are designed to target only large consumers\cite{wang2010demand}\cite{hu2007implementation}\cite{zhou2010overview}\cite{andrews2009china}. However, we have demonstrated that a consumer's MCI is determined by their demand profile rather than demand level. Consumers' MCIs are not necessarily positively correlated with their demand levels. Thus, "demand-level based" prices or "manage large consumers only" principle are efficient in reducing system cost only if consumers' MCIs happen to be positively correlated with their demand levels in a DS.

In our studied distribution system, consumers' MCIs are not positively correlated with their demand levels. In Fig.~\ref{fig:MFIMRvsUsage} , we plot the correlation between consumers' demand levels and their $MFI_{MR}$. Many large consumers have smaller $MFI_{MR}$ than small consumers. We note that $MFI_{MR}$'s variations significantly increase when consumers' demand levels exceed a certain level. The correlations between consumers' demand levels and their $MFI_{ER}$ and $MFI_{DR}$ have the similar behaviors as shown in the Appendix. Therefore, it is highly inefficient to implement "demand-level based" policies or "manage large consumers only" principle in this DS.

\begin{figure}
\centering
\includegraphics[width = 0.45\textwidth,height = 2in]{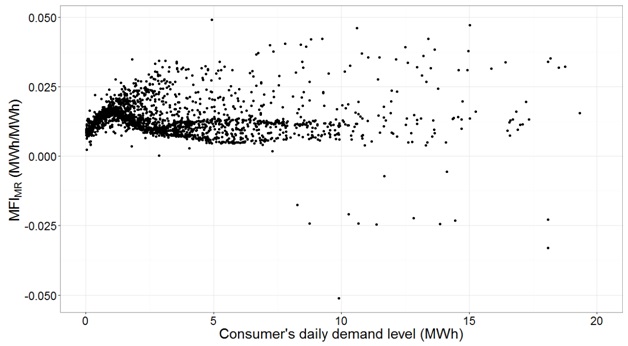}
\caption{Correlation between consumer's $MFI_{MR}$ and demand level}
\label{fig:MFIMRvsUsage}
\end{figure}

\subsection{Assess the efficiency of the implemented retail-price scheme}
The price scheme in this distribution system is complicated. We summarize the price scheme in the Appendix. In general, commercial and industrial consumers were charged a time-of-use tied price. Residential and public-service consumers had a flat rate.

We do not have detailed information about the daily cost structure of this distribution system because China does not have a wholesale market. Therefore, we use the $DOC$ index defined in Sec. ~\ref{sec:criteriaforprice} to examine the efficiency of the current retail price scheme. 

We calculate each consumer's daily average rates. Because average rates are distributed continuously between $0.35$ Yuan/kWh to $0.87$ Yuan/kWh, we equally split $[0.35,0.87]$ to 36 segments, which is the number of consumers' profile types, and cluster consumers into the same average-daily-rate type if their average-daily rates are in the same segment. By comparing the two clusters, We calculate the daily $DOC$ index of the current retail price scheme and summarize the monthly average of each weekday in Fig.~\ref{fig:DOC}. 
\begin{figure}
\centering
\subfigure[Daily DOC of the current retail price]{
\includegraphics[width = 0.4\textwidth]{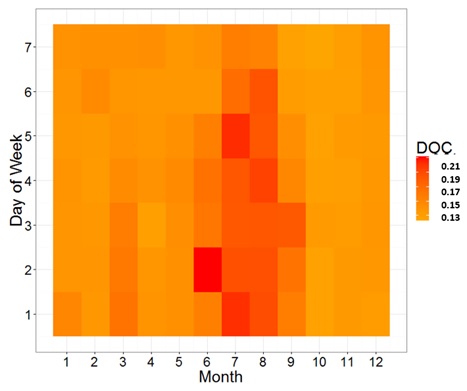}
\label{fig:DOC}
}
\subfigure[Daily Dt of the current retail price]{
\includegraphics[width = 0.4\textwidth]{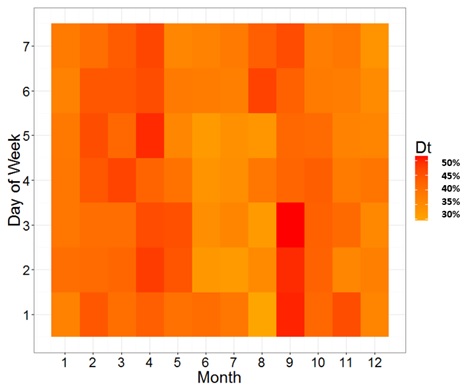}
\label{fig:Dt}
}
\label{heatmap}
\caption{Economic inefficiency and distortion of the current retail price}
\end{figure}

The retail price scheme is very different from an economically optimal price. The highest daily $DOC$ index in these 11 months is only $0.3825$, which is far less than $1$. Therefore, the randomness in daily average rates can only limitedly be reduced if we know a consumer's daily profile on a particular day. We can conclude that the current retail price scheme is really inefficient even though we do not know this DS's daily cost structures.

The extent of retail price inefficiency varies by month. As shown in Fig.~\ref{fig:DOC}, the retail price scheme was most efficient during weekdays in summer. 

We assume the daily system cost is monotonically increasing with $MFI_{MR}+MFI_{ER}+MFI_{PD}$ and calculate the daily $Dt$ index. Under our assumptions, the current retail price scheme created significant distortion in this distribution system. $40$\% to $60$\% consumers pay a higher rate than their neighbors whose $MCIs$ are higher.

\section{Conclusion and discussion: manage consumers according to their marginal system impacts}\label{sec:conclusion}
In this paper, we explain how a distribution system's daily cost and associated emissions are affected by consumer daily demand profiles. We clarify that a consumer's marginal system impacts are determined by their demand profile rather than their demand level. Thus, profile-based consumer clustering clusters consumers according to their marginal system impacts. We also demonstrate that a consumer's daily marginal system cost impact is the economically optimal daily average rate. If we design a profile price scheme whose rate is equal to consumer's daily marginal system cost impact, the profile price is equivalent to real-time pricing and is an economically optimal price scheme. 

We indicate that clustering demand profile is identical to cluster consumers according to their marginal system impacts. When consumers with various demand profiles have different marginal impacts on the system daily cost, demand profile clustering is identical to daily-average-rate clustering only if the implemented price is economically optimal, . Based on these theoretical analyses, we develop a criteria system to evaluate the economic efficiency of an implemented retail price scheme in a distribution system. Our criteria system can examine the extent of a retail price scheme's efficiency even if we do not have information about the distribution system's daily cost structure. These criteria can also be used to target consumers who overpay or underpay for their electricity usage. 

We analyze data from a real distribution system in China and examine consumer marginal system impacts and the efficiency of the retail price scheme there. The empirical results deepen our insights and understanding about the multi-dimensional marginal system impacts of consumers. The results strongly suggest that we redesign current retail price schemes and rethink the principles used to direct demand-response project designs.

In general, the design for retail price or demand-response projects needs to be smarter by considering the following three issues.

First, it is very valuable for a distribution system operator or utility to accurately target and manage "particular days" in which a profile type's MFIs are extremely low or high. Fig. ~\ref{fig:boxplot} demonstrates that there are a small number of "particular days" for each profile type. Thus, every consumer will be affected in very few days if system operators or utilities target and manage "particular days". However, in these small number of "particular days", it is either very valuable or extremely expensive to serve the corresponding type of consumers. Therefore, distribution systems can significantly improve economic and environmental performance by developing demand management projects for "particular days" for each profile type.

Second, the large within-type variations of MFIs remind us that demand management, including retail pricing design, must make a trade-off between long-term mean and variation of consumer's system impacts. In the studied distribution system in our research, consumers that have small MFIs usually have stable MFIs over days. For these consumers, their retail rate can be stable in the whole year. In contrast, complicated designs for retail price rates or demand-response projects are necessary to manage consumers whose MFIs are large and significantly varying.

Finally and most importantly, the "marginal system impact based" principle should replace the "demand level based" principle in order to optimize demand management. Consumers' "marginal system impacts" reflects which consumption behaviors are expensive to serve and need to be the foundation of retail price designs.

\bibliographystyle{IEEEtran}
\bibliography{PESLargeUserVersion}

\end{document}